%% file: schelling-main.tex
\documentclass[11pt]{article}

\usepackage{ifthen}
\usepackage{mdwlist}
\usepackage{fullpage}
\usepackage{amsmath,amsfonts,amsthm}
\usepackage{bm}
\usepackage{enumitem}
\usepackage{graphicx}
\usepackage[ruled,linesnumbered]{algorithm2e}
\usepackage{xspace}

\usepackage[pdftex]{color}
\newcommand{\nsi}[1]{{\textcolor{blue}{\bf [Nicole: #1]}}}

\newcommand{\rdk}[1]{{\textcolor{red}{\bf [Bobby: #1]}}}
\definecolor{darkblue}{rgb}{0,0.08,0.45}

\newcommand{\ccount}{{\chi}}

\input{preamble}

\newtheorem{theorem}{Theorem}

\newtheorem{definition}{Definition}

\newcommand{\ignore}[1]{}

\newcommand{\Prb}[1]{\mathbb{P}\left[#1\right]}

\title{An Analysis of One-Dimensional Schelling Segregation}
\author{
Christina Brandt\thanks{Department of Computer Science, Stanford University.  Supported by the National Science Foundation Graduate Research Fellowship under Grant No. DGE-1147470.  Email: {\tt cbrandt@stanford.edu}} 
\and
Nicole Immorlica\thanks{Department of Electrical Engineering and Computer Science, Northwestern University.  Supported in part by NSF awards CCF-1055020 and SMA-1019169, a Microsoft New Faculty Fellowship, and an Alfred P.~Sloan Foundation Fellowship.  Email: {\tt nickle@eecs.northwestern.edu}} 
\and 
Gautam Kamath\thanks{Department of Computer Science, Cornell University.  Email: {\tt gck43@cornell.edu}}
\and 
Robert Kleinberg\thanks{Department of Computer Science, Cornell University.
Supported in part by NSF awards CCF-0643934 and AF-0910940, AFOSR grant FA9550-09-1-0100, a Microsoft Research New Faculty Fellowship, an Alfred P.~Sloan Foundation Fellowship, and a Google Research Grant.  Email: {\tt rdk@cs.cornell.edu}}
}
\date{}

\begin{document}


\maketitle

\begin{abstract}
\input{abstract}
\end{abstract}





\section{Introduction}
\label{sec:intro}
\input{introduction-condensed}

\section{Preliminaries}
\label{sec:prelim}
\input{preliminaries}

\section{Analysis}
\label{sec:analysis}
\input{analysis}

\subsection{Bounding average run length}
\label{subsec:runlength}
\input{runlength}

\subsection{Bounding number of unhappy elements}
\label{subsec:unhappies}
\input{unhappies-summ}

\section{Open Questions}
\input{openq}

\section{Acknowledgements}
\input{ack}

\bibliographystyle{abbrv}
\bibliography{biblio}

\appendix
\section{Deferred Proofs}
In this appendix we present proofs that were deferred
from Section~\ref{sec:analysis}.
\label{sec:runapp}

\input{runapp}

\subsection{Applying Wormald's Technique}
\label{sec:wormald}

\input{wormald}

\end{document}

%% file: preamble.tex
\def\eps{\epsilon}

\def\to{\rightarrow}
\def\E{{\bf E}}

\def\cale{{\cal E}}

\newcommand{\comment}[1]{}

\newcommand{\prob}[2][]{\text{\bf Pr}\ifthenelse{\not\equal{}{#1}}{_{#1}}{}\!\left[#2\right]}
\newcommand{\expect}[2][]{\text{\bf E}\ifthenelse{\not\equal{}{#1}}{_{#1}}{}\!\left[#2\right]}

\newtheorem{thm}{Theorem}[section]

\newtheorem{lem}[thm]{Lemma}
\newtheorem{prop}[thm]{Proposition}
\newtheorem{claim}[thm]{Claim}

\newcommand{\bg}[1]{\medskip\noindent{\bf #1}}



%% file: abstract.tex
We analyze the Schelling model of segregation in which a society of $n$ individuals live in a ring.  Each individual is one of two races and is only satisfied with his location so long as at least half his $2w$ nearest neighbors are of the same race as him.  In the dynamics, randomly-chosen unhappy individuals successively swap locations.  We consider the average size of monochromatic neighborhoods in the final stable state.  Our analysis is the first rigorous analysis of the Schelling dynamics.  We note that, in contrast to prior approximate analyses, the final state is nearly integrated: the average size of monochromatic neighborhoods is independent of $n$ and polynomial in $w$.  

%% file: introduction-condensed.tex

In 1969, economist Thomas Schelling introduced a landmark model of racial segregation.
Elegantly simple and easy to simulate, it provided a persuasive 
 explanation of an unintuitive result: that local behavior can cause global effects that are undesired by all~\cite{schelling_models_1969}.  In Schelling's model, individuals of two races, denoted $x$ and $o$, are placed in proximity to one another, either in a line (the one-dimensional model) or in a grid (the two-dimensional model).  This represents a mixed-race city, where individuals of different races live in close proximity.  Individuals are satisfied if at least a fraction $\tau$ of the other agents in a small local neighborhood around them are of the same type.  Unhappy agents can move locations, either by inserting themselves into new positions or exchanging locations with other agents.  Schelling showed via small simulations that global segregation can occur even when no individual prefers segregation.  In his experiments, he found that on average, an individual $i$ with $\tau=\tfrac12$ ended up in a significantly more segregated neighborhood with approximately $80\%$ of $i$'s neighbors with $i$'s type.


The striking contrast between individual preferences and global effects captured the imaginations of sociologists, econ\-omists and physicists.  Schelling's model eloquently argues that while all individuals in a community may prefer integration, the global consequence of their actions may be complete segregation.  Empirical evidence, found through surveys and statistics of segregated communities, indicates that the effects of local organization seen in Schelling's spatial proximity model may also lead to real-world segregation, whether by ethnicity~\cite{
alba_minority_1993, 
clark_geography_2008, 
clark_understanding_2008,
emerson_does_2001,
farley_continued_1993, 
massey_migration_1994, 
trappenburg_segregation_2006,
van_der_laan_bouma-doff_involuntary_2007,
vargas_causes_2006},
religion~\cite{don-yehiya_religion_1999},
or socioeconomic factors~\cite{
benard_wealth_2007, 
benenson_schelling_2009, 
mccrea_explaining_2009}. 

However, it is surprisingly difficult to analytically prove or even rigorously define the segregation phenomenon observed qualitatively in simulations.  Much of this difficulty lies in the fact that the dynamics may converge to a variety of states (complete segregation, complete integration, and various partially segregated states), and so the underlying Markov chain does not have a well defined unique stationary distribution.  Prior work~\cite{young_individual_2001,zhang_dynamic_2004,zhang_residential_2004,zhang_tipping_2011} circumvents this difficulty by introducing perturbations in the dynamics, allowing individuals to perform detrimental actions with vanishingly small probability.  The research then analyzes the degree of segregation in stochastically stable states, finding generally that as time approaches infinity, complete segregation is inevitable. 


We instead analyze the one-dimensional segregation dynamics directly, providing the first rigorous analysis of an unperturbed Schelling model.
Our model considers a society of $n$ individuals arranged in a ring network.  We start from a random initial configuration in which each individual is assigned type $x$ or $o$ independently and uniformly at random.  We parameterize the neighborhood size in the Schelling dynamics by $w$ and then consider the model of dynamics under which random pairs of unhappy individuals of opposite types trade places in each time step.  In the initial configuration, the average size of a monochromatic neighborhood is constant.  We show that with high probability, the dynamics converge to a stable configuration.  
In stark contrast to previous analyses of approximate dynamics, we prove that once the dynamics converge, most individuals reside in nearly integrated neighborhoods; 
that is, the average run length in the final configuration is independent of $n$ and only polynomial in $w$.  Thus, contrary to the established intuition, the local dynamics of Schelling's model do not induce global segregation in proportion to the size of the society but rather induce only a small degree of local segregation.  This does not contradict empirical studies, which are often performed over small local populations, but indicates that the results found in smaller communities do not necessarily generalize to larger populations.  Our results are in accord with empirical studies of residential segregation in large populations~\cite{clark1986residential,farley_continued_1993,white1986segregation} which consistently find amounts of spatial autocorrelation that are intermediate between purely random and completely segregated configurations.

\paragraph{Our techniques}
As is common in the analysis of non-stationary random processes on graphs
and other discrete structures, our process can be analyzed by defining state 
variables whose expected values at any time can easily be estimated
based on their values in the preceding time step.  
The time-evolution of these state
variables, after a suitable renormalization, can be approximated with
a continuous-time vector valued process that is not random at all,
but satisfies a differential equation obtained from the leading-order
terms in the aforementioned relations between successive
time steps.  
Wormald~\cite{wormald-annals,wormald} supplied a general theorem that
rigorously justifies the use of such differential equation 
approximations in a wide range of 
cases. 
In the context of Schelling segregation, the natural parametrization
of the state is a vector with infinitely many components,
reflecting the normalized frequency of each finite string of
$x$'s and $o$'s as one looks at all labeled subintervals of the 
ring.  
There are two reasons why it is not 
straightforward to apply Wormald's technique
in this setting.  
First, the system of differential equations
has infinitely many variables and infinitely
long dependency chains;
 Wormald's technique only applies
when there is no infinite sequence of distinct variables
$y_1,y_2,\ldots$ such that the derivative of $y_i$ depends
on the value of $y_{i+1}$ for all $i$.  Second, the 
continuous-time system is much too complex to analyze its solution
directly, so it is unclear how to extract any meaningful
bounds on the eventual distribution of run-lengths 
by working directly with the differential equation.

We circumvent the first
difficulty by partitioning the ring into bounded-sized pieces
using small separators, and analyzing a different random
process taking place on the partitioned graph.  A coupling
argument shows that after running the bounded-size process for
$O(n)$ steps, with high probability its state vector
closely approximates
the relevant components of the original state vector.
The new process has only finitely many state variables,
so Wormald's technique is applicable.
We believe that this technique of partitioning along small
separators to remove weak long-range dependencies may be 
of use in other applications of Wormald's technique that
suffer from infinitely long dependency chains.

To overcome the second difficulty, rather than analyzing
the differential equation solution directly, we focus on 
the presence of a particular configuration that we call
a \emph{firewall}: a string of $w+1$ consecutive individuals
of the same type.  These configurations are stable for the
segregation process: the neighborhood of each element in a firewall contains 
at least $w$ elements of its own type, so for $\tau=\tfrac{1}{2}$, 
once a firewall is formed, it cannot be subsequently broken.  
Thus, to prove that a typical site never belongs to a
monochromatic run of superpolynomial length, it suffices
to show that firewalls of both colors form within a distance of
$\operatorname{poly}(w)$ on both sides of a given site, with
high probability.  We prove this by defining a special 
type of configuration called a \emph{firewall incubator},
which occurs with probability $\Omega(1)$ 
at any given site in the initial configuration, and has
probability $1/\operatorname{poly}(w)$ of developing into
a firewall wherever it occurs in the initial configuration.
The proof of the latter fact depends on a symmetry 
property of the differential equation: it is invariant
under the $\mathbb{Z}/(2)$ action that exchanges $x$'s and $o$'s,
and therefore the fixed points of this $\mathbb{Z}/(2)$ action
are an invariant set for the differential equation.
This symmetry allows us to reduce our problem to the
analysis of a simpler process in which every step
consists of selecting a single random site and changing
its color if it is unhappy.

\subsection{Related Work}
Schelling's proximity model of segregation, first introduced in 1969~\cite{schelling_models_1969}, inspired significant research into understanding the dynamics of prejudice and self-isolation of communities.   Schelling defined a one-dimensional proximity model in which a community is represented by individuals placed next to one another in a line.  Each individual's neighborhood is composed of the $w$-closest elements on either side.  In Schelling's simulations, he let $w=4$, so each element's neighborhood contained $8$ elements: the $4$ nearest elements on each side.  The satisfaction of the individuals in the line was parametrized by a single global tolerance value, $\tau$.  For any individual $i$, if less than $2w\tau$ of $i$'s neighbors are of $i$'s type, then $i$ was considered unhappy.  Unhappy individuals were chosen one at a time and inserted into the nearest position where their tolerance requirement would be satisfied.  In small simulations (with $N=60$ agents), he found that even when $\tau\leq \frac{1}{2}$, most agents ended up in fully segregated neighborhoods~\cite{schelling_models_1969}.  Schelling continued to build upon this initial work, introducing a two-dimensional model where agents were placed on a partially empty grid.  Unhappy agents could move to empty positions on the grid where they would become happier.  Again, via small simulations, Schelling concluded that even with $\tau\leq \frac{1}{2}$, segregation was essentially inevitable~\cite{schelling_dynamic_1971, schelling_strategy_1980}.  Significant research has been done to extend the Schelling spatial proximity model~\cite{
bardea_back_2011,
benard_wealth_2007, 
benenson_schelling_2009,
dalasta_statistical_2008,
gerhold_limit_2006,
grauwin_dynamic_2009,
henry_emergence_2011,
singh_schellings_2007} 
and analyze it~\cite{
bardea_back_2011,
dalasta_statistical_2008,
grauwin_dynamic_2009,
pancs_schellings_2007,
zhang_dynamic_2004, 
zhang_residential_2004,
zhang_tipping_2011,
young_individual_2001}.

Young was the first to present a more rigorous analysis of the model~\cite{young_individual_2001}.  He utilized the technique of stochastic stability, developed in evolutionary game theory, to analyze a variant of the one-dimensional Schelling spatial proximity model.  In Young's version of the model, agents are placed on a ring and have neighborhood width $w=1$ and tolerance $\tau = \frac{1}{2}$.  Given values $0 < a < b < c$ and $0 < \eps < 1$, at each time step, two agents are chosen at random and trade places with probability $1$ if the trade makes both happy, $\eps^a$ if one changes from unhappy to happy and the other vice-versa, $\eps^b$ if both are initially happy and one becomes unhappy, and $\eps^c$ if both change from happy to unhappy.  Young utilizes the concept of stochastically stable states in a coordination game to analyze the Markov chain and its possible perturbations, concluding that with high probability as $\eps \to 0$, total segregation will result.  Further generalizations and variants of this model have been rigorously analyzed by Zhang~\cite{zhang_dynamic_2004, zhang_residential_2004, zhang_tipping_2011}, Barde~\cite{bardea_back_2011}, Dall'Asta et al.~\cite{dalasta_statistical_2008}, Pancs and Vriend~\cite{pancs_schellings_2007}, Grauwin~\cite{grauwin_dynamic_2009}, and others.

Our model differs from previous research and returns to a model closer to Schelling's original in that agents do not take utility-decreasing moves and are not analyzed in terms of bounded neighborhoods.  We find that these simple differences lead to substantially altered results.

Stochastic models of residential segregation fit within the broader
scope of social science models that study the aggregate behavior of
large networks of agents each individually executing very simple,
myopic, often randomized procedures to update their behavior in
response to the behavior of their neighbors.  Highlights of this
line of work include evolutionary game theory and the study
of stochastically stable states~\cite{foster-young,kmr,young93}, the
analysis of coordination games played on 
networks~\cite{ellison,montanari-saberi},
analysis of repeated best-response dynamics in large
games~\cite{blume95,gmv}, and research drawing explicit
parallels between statistical mechanics models and the
dynamics of large-population games~\cite{blume93}.

Finally, our paper belongs to a long line of papers in theoretical computer science
and discrete probability that apply differential equations to analyze
the dynamics of non-stationary random processes.
An early application of this technique is Karp and 
Sipser's~\cite{karp-sipser} analysis of a random greedy matching 
algorithm in random graphs.  Differential equations have 
also been applied to 
analyze algorithms for random {\sc $k$-Sat} 
instances~\cite{achlioptas,chen-franco,CGHS04,frieze-suen}, 
study
component sizes in random graphs~\cite{aldous,molloy-reed},
and to analyze ``Achlioptas processes'' in which edges are 
added to an initially empty graph by an algorithm selecting
among a bounded number of random choices~\cite{ADS-explosive,bohman-frieze,FGS-embracing,riordan-warnke,spencer-wormald}.
Wormald~\cite{wormald-annals,wormald} provides very general
conditions under which differential equation
approximations such as these are guaranteed to have
$o(n)$ additive error in the large-$n$ limit.

%

%% file: preliminaries.tex
We consider a society of $n$ individuals.  An individual's type is either $x$ or $o$, with a probability $p$ of being $x$ and $(1-p)$ of being $o$.  Here we take $p=1/2$.  Individuals live in a ring network represented by an $n$-node cycle.  At each point in time, there is a bijective mapping between individuals and nodes.  Each individual lives in a {\it neighborhood}, defined to be his $2w+1$ nearest neighbors (including himself) for a parameter $w \ll n$, i.e., the neighborhood of an individual at node $i$ in the ring consists of the set of individuals at nodes $\{(i-w)\mod n,\ldots,(i+w)\mod n\}$.  The parameter $w$ is called the {\it window size}.

We say an individual is {\it happy} if at least a $\tau$ fraction of his neighbors are of the same type as him.  The parameter $\tau$ is called the {\it tolerance parameter} and here is assumed to be $1/2$.  At any given time step, two individuals are chosen uniformly at random and swap nodes according to the following rules.  If both individuals are unhappy, are of opposite types, and would therefore be happy\footnote{Here we are applying the assumption that $\tau=1/2$.  When $\tau>1/2$ an unhappy individual may remain unhappy after swapping with an oppositely colored unhappy individual.  In such a case, a natural modeling assumption is that the swap still takes place as long as each individual has at least as many neighbors of the same type at his new location as at his former location.} in the other's node, then they swap nodes.

A {\it block} is a sequence of adjacent sites.  A {\it run} is a block whose nodes are identically labeled.  (Throughout the paper, we use the term {\em label} to denote the type of the individual living at a given node.)

A key observation underlying our analysis of the Schelling process is that individuals in large enough runs never have an incentive to move. 
Define a {\it firewall} as a run of length at least $w+1$; a firewall
is either an {\em $x$-firewall} or an {\em $o$-firewall} depending on
the labels of its nodes.
For a segregation process with tolerance $\tau=\tfrac{1}{2}$,  since individuals living in a firewall have at least $w$ adjacent neighbors of the same type, all elements in the firewall are happy and will remain so, independent of the labels of the nodes around the firewall.  
Therefore, once a firewall is created, no individual in the firewall will move and the configuration will remain stable for the remainder of the process.

In addition, the existence of at least one firewall in the initial configuration also guarantees that the process will eventually reach a \emph{frozen configuration} in which no further swaps are possible.
\begin{prop}\label{prop:frozenstate}
Consider the segregation process with window size $w$ 
on a ring network of size $n$.  For any fixed $w$,
as $n \to \infty$, the probability that the process 
eventually reaches a frozen configuration converges to 1.
\end{prop}
\begin{proof}
A potential function that verifies this fact is the number of 
individuals belonging to firewalls.  Let $S_0(t)$ denote the set of 
individuals belonging to firewalls at time $t$, and let $S_1(t)$ denote
its complement.  As long as $S_0(t)$ and $S_1(t)$ are both nonempty, 
there will be an individual $a \in S_1(t)$ neighboring an individual 
$b \in S_0(t)$.  These two must be oppositely labeled, as otherwise $a$
would belong to the same firewall as $b$.  Individual $a$ must
be unhappy: assuming w.l.o.g.\ that $b$ lives to the right of 
$a$ and that the label of $a$ is $x$, then all of $a$'s neighbors 
on the right are labeled $o$, and at least one of $a$'s neighbors
on the left is labeled $o$ (as otherwise $a$ itself would belong
to an $x$-firewall), and hence $a$ is unhappy.  If there are
unhappy individuals of both labels, then there is a positive probability
that $a$ will swap with an oppositely labeled unhappy individual
and that individual will then join a firewall, increasing the
potential function.  If all of the unhappy individuals are labeled
$x$, then all of the $o$-type individuals are already living in
firewalls and the configuration is already frozen.  

Thus, we have defined an integer-value potential function, taking
values between $0$ and $n$, that always has positive probability
of increasing unless either the configuration is already frozen,
or the initial configuration contained no firewalls.  Finally,
the probability that the initial configuration contains no 
firewalls is certainly $o(1)$: partition the ring into $n/(w+1)$
blocks, each of which independently has probability 
$2^{-w}$ of being a firewall in the initial configuration.
The probability that none of them are firewalls is 
$(1-2^{-w})^{-n/(w+1)}$, which is $o(1)$ as $n \to \infty$.
\end{proof}

Firewalls are therefore stable, segregated configurations which 
guarantee the eventual termination of the process.

%% file: analysis.tex
In this section we prove bounds on the run-length distribution
of the segregation model.  
The main idea behind our analysis is to show that 
firewalls 
 occur fairly frequently.  We will show that
for any site on the ring, with high probability, the process
will eventually form firewalls of both colors on both sides of 
the site, within its $\mathrm{poly}(w)$ nearest neighbors.  To do so,
we define a type of configuration that we call a \emph{firewall
incubator} and we show that it occurs reasonably frequently in
a random 2-coloring of the ring: every site has probability $\Omega(1)$
of belonging to a firewall incubator in the initial
configuration.  The main part of the proof is devoted to showing
that firewall incubators are reasonably likely to develop into
firewalls; the probability is at least $1/\mathrm{poly}(w)$.  
To prove this, it is easier to analyze a related stochastic
process in which one step consists of choosing
a single site and reversing its color if it is unhappy.
The comparison with this process is justified if the 
ratio of unhappy $x$'s to unhappy $o$'s is equal to
$1 \pm o(1)$ throughout the time interval of interest;
we prove that this is so, with high probability, by an application
of Wormald's differential equation technique.

%% file: runlength.tex
\newcommand{\sat}[1]{t^*_#1}
\newcommand{\sato}{t^*}
\newcommand{\sattm}{satisfaction time\xspace}
\newcommand{\sumseq}{transcript\xspace}
\newcommand{\psumseq}{pseudo-transcript\xspace}

As discussed above, to analyze the formation of firewalls
we will first define  a structure
called a \emph{firewall incubator} that has a reasonable probability
of becoming a firewall in the long run.  An incubator is a region with substantially more $x$ sites than $o$ sites (or vice versa).  In such regions, the minority individuals are unhappy and will continue to move out unless nearby neighborhoods have developed an opposite bias.  Using a random walk analysis, we will argue that it is reasonably likely that all the minority individuals move out before this happens, and so the region turns into a firewall.

\paragraph{The Birth of an Incubator}  
A firewall incubator, defined formally in Definition~\ref{def:incubator} below, consists of a sequence of blocks: two \emph{defender} blocks flanking an \emph{internal} block.
The blocks of size $w$ on either side of the firewall incubator are called \emph{attacker} blocks, and they play a key role in our analysis.  Defender and internal blocks are the regions that will potentially become firewalls.  The attacker blocks are the nearby neighborhoods that might impede the process of the defenders becoming firewalls.  The internal blocks are also biased in the same direction as the defender blocks and so help guarantee that the defender is not attacked from both sides.  

To specify our exact requirements on the biases, we associate a sign, $+1$ or $-1$, with $x$ and $o$, respectively, and define the \emph{x-bias} $\beta_t(i)$ of a node $i$ at time $t$ to be the sum of the signs of the $w$-closest nodes on either side of the element and the sign of the element itself. 
We will write $\beta(i)$ when the time is clear from context.
The $x$-bias directly expresses the element's happiness: an $x$-element $i$ is happy if and only if $\beta(i)>0$ and an $o$-element $j$ is happy and only if $\beta(j)<0$.\footnote{Since $\beta(i)$ is the sum of $2w+1$ labels, it is always odd, and so one of these strict inequalities must hold.}  
\begin{definition}
\label{def:incubator}
A \emph{firewall incubator} is a block $F$ made up of three consecutive
blocks $D_L,I,D_R$ 
(called left defender, internal, and right defender, respectively)
such that:
\begin{enumerate*}
\item $D_L$ and $D_R$ have exactly $w+1$ nodes;
\item $\beta_0(i) > \sqrt{w}$ for all $i \in F$;
\item The minimum $x$-bias in $D_L$ occurs at its left endpoint, and the
minimum $x$-bias in $D_R$ occurs at its right endpoints.
\end{enumerate*}
The blocks of length $w$ immediately to the left and right of $F$ are
denoted by $A_L, A_R$ and are called the left and right attackers.
\end{definition}
For an example of an incubator, see Figure~\ref{fig:incubator}.
\begin{figure}[t]
  \centering
\includegraphics[width=4.5in]{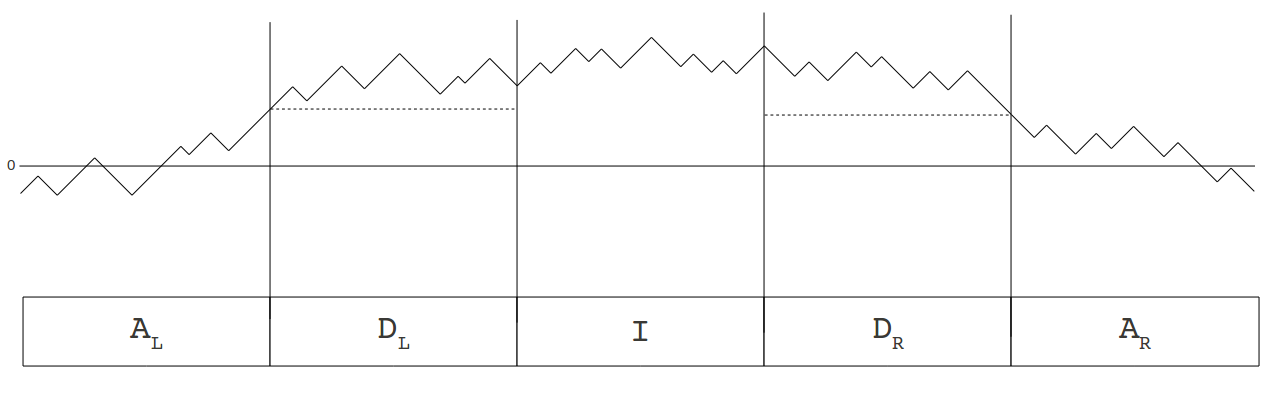}
  \caption{A firewall incubator, surrounded by left and right attacking blocks. Height indicates $x$-bias of a position.}
\label{fig:incubator}
\end{figure}
Firewall incubators occur fairly frequently in the initial configuration.
In fact, Proposition~\ref{prop:incubators-happen} below 
proves that every block of size $6w$ in the initial configuration has 
probability $\Omega(1)$ of containing a firewall incubator.
The proof consists of partitioning the block (along with its neighboring
attacker blocks) into sub-blocks of
length $w$ and showing that the 
property $\forall i \; \beta_0(i)>\sqrt{w}$ is implied by 
conditions on the partial sums of the sign sequence in
each sub-block.  These conditions are verified to hold 
with constant probability, using the reflection principle
and the central limit theorem.  Finally, the proof shows
that by choosing the left and right endpoints of the 
incubator to be the nodes with minimum bias in the leftmost
(resp.\ rightmost) sub-block, with constant probability
these nodes also have the minimum bias in $D_L, D_R$,
respectively.

\begin{prop}
\label{prop:incubators-happen}
Let $r$ be an integer such that $6 \leq r < \frac{n}{w} - 2$.
For any sequence of 
$rw$ consecutive nodes, the probability that a 
uniformly random $\{x,o\}$-labeling of the nodes 
contains an $x$-firewall incubator that
starts among the leftmost $w$ nodes and ends among
the rightmost $w$ nodes is at least $c^r$, where 
$c > 0$ is a constant independent of $r,w,n$.
\end{prop}

\newcommand{\labeling}{\lambda}

To prove the theorem, we start by defining some notation.
As above, we associate a value of $+1$ to a node labeled
with $x$ and $-1$ to a node labeled with $o$.  If $B$ is 
any sequence in $\{x,o\}^k$, we use 
$\ccount_j(B)$ to denote the sum of 
the first $j$ associated signs, for $0 \leq j \leq k$. 
We also use $\ccount(B) = \ccount_k(B)$ to denote the sum of all signs
in $B$, and $\ccount_{-j}(B) = \ccount(B) - \ccount_{k-j}(B)$ to denote 
the sum of the final $j$ signs.

\begin{definition} \label{def:promoting}
A sequence $B \in \{x,o\}^{w}$ is \emph{$x$-promoting}
if $\ccount(B) \geq 5 \sqrt{w}$ and for all $j = 1,\ldots,w$,
$\ccount_j(B) > - 2 \sqrt{w}$ and $\ccount_{-j}(B) > - 2 \sqrt{w}$.
\end{definition}
\begin{lem} \label{lem:promoting-1}
The probability that a uniformly random sequence $B \in \{x,o\}^{w}$ 
is $x$-promoting is $\Omega(1)$.
\end{lem}
\begin{proof}
For $k>0$, let $\cale^L_k, \cale^R_{k}$ denote the events
\begin{align*}
\cale^L_k &= \{ \exists j \; \ccount_j(B) \leq -k \} \\
\cale^R_k &= \{ \exists j \; \ccount_{-j}(B) \leq -k \}
\end{align*}
The probabilities of these events can be calculated using
the reflection principle~\cite{feller}, which can be
phrased in this context as follows:
{\em if $B$ is a uniformly random element of
$\{x,o\}^{w}$ then for all $k>0$}
$$
\Pr(\cale^L_k) = 
\Pr(\ccount(B) \leq -k) + \Pr(\ccount(B) < -k).
$$
By symmetry, the right side is equal to 
$\Pr(\ccount(B) \leq -k) + \Pr(\ccount(B) > k)$.
Using Markov's inequality, and the fact that $\ccount(B)$ is a sum of 
$w$ independent random signs and hence
$\E[(\ccount(B))^2] = w$, we now obtain
\begin{align*}
\Pr(\cale^L_k) & \leq
\Pr((\ccount(B))^2 \geq k^2) \leq \frac{w}{k^2},
\end{align*}
for all $k > 0$.  In particular, the right side is
$1/4$ when $k = 2\sqrt{w}$. 

It is easy to see that whenever $a<b$ are two numbers
such that the events $\ccount(B) = a$ and $\ccount(B)=b$
have positive probability, the inequality
\begin{equation} \label{eq:promoting-1}
\Pr(\cale^L_k \mid \ccount(B)=a) \geq 
\Pr(\cale^L_k \mid \ccount(B)=b)
\end{equation}
holds.  One way to see this is to observe that we can
obtain a random
sample from the conditional distribution of $B$ given
$\ccount(B)=a$ by the following procedure: first draw
a random sample from the conditional distribution of $B$
given $\ccount(B)=b$, then select a uniformly random 
set of $\frac{b-a}{2}$ occurrences of $x$ in the sequence 
and change each of them to $o$.  The second stage of the 
sampling does not increase any of the partial sums
$\ccount_j(B)$, so if $\cale^L_k$ held in the first stage
of the sampling, then it continues to hold after the second
stage.
Inequality~\eqref{eq:promoting-1} implies that for any $b$
such that the events $\ccount(B)<b, \ccount(B) \geq b$ both
have positive probability,
\begin{equation} \label{eq:promoting-2}
\Pr(\cale^L_k \mid \ccount(B) < b) \geq
\Pr(\cale^L_k \mid \ccount(B) \geq b).
\end{equation}
Since the unconditional probability of $\cale^L_k$ is 
a weighted average of the left and right sides, we obtain
\begin{equation} \label{eq:promoting-3}
\Pr(\cale^L_k \mid \ccount(B) \geq b) \leq
\Pr(\cale^L_k) \leq \frac{w}{k^2},
\end{equation}
for all $b,k$.  By symmetry, $$\Pr(\cale^R_k \mid \ccount(B) \geq b) =
\Pr(\cale^L_k \mid \ccount(B) \geq b).$$
By the Central Limit Theorem,
$$
\lim_{w \to \infty} \Pr(\ccount(B) \geq 5 \sqrt{w}) = 
\frac{1}{\sqrt{2 \pi}} \int_5^{\infty} e^{-x^2/2} \, dx
$$
and therefore there is an absolute constant $c_0$ such
that $\Pr(\ccount(B) \geq 5 \sqrt{w}) > 2 c_0$ for all 
$w > 25$.  (The restriction to $w>25$ is necessary so that
$w \geq 5 \sqrt{w}$.)  Now we find that for $b=
5 \sqrt{w}$ and $k=2 \sqrt{w}$,
\begin{align*}
\Pr(\mbox{$B$ is $x$-promoting}) & =
\Pr \left( \ccount(B) \geq b \, \wedge \,
\overline{\cale^L_k} \, \wedge \,
\overline{\cale^R_k} \right) \\
& \geq
\Pr(\ccount(B) \geq b) \, \cdot \,
\left[ 
1 - 2 \Pr(\cale^L_k \mid \ccount(B) \geq b) \right] \\
& \geq
2 c_0 \left[ 1 - 2 \left( \frac{w}{k^2} \right) \right] = c_0.
\end{align*}
\end{proof}

Now we proceed to the proof of Proposition~\ref{prop:incubators-happen}.
\begin{proof}[Proof of Proposition~\ref{prop:incubators-happen}]
Given a sequence of $rw$ consecutive nodes, for 
$6 \leq r \leq \frac{n}{w}-2$, partition it into
$r$ blocks $B_1,B_2,\ldots,B_r$ each containing $w$
consecutive nodes.  Let $B_0, B_{r+1}$ denote the blocks
of $w$ nodes immediately preceding $B_1$ and 
immediately following $B_r$, respectively.  
If $\labeling = \labeling_{00}$ is any labeling of 
the nodes in $B_0,B_1,\ldots,B_{r+1}$, then let
$\labeling_{01},\labeling_{10},\labeling_{11}$ respectively
denote the labelings obtained from $\labeling$ by reversing
the ordering of the labels of the first $4w$ nodes, the last
$4w$ nodes, or both sets of nodes.

With
probability at least $c_0^{r+2}$, the 
$r+2$ blocks that constitute $\labeling_{00}$ are all
$x$-promoting.  The reverse of an $x$-promoting
sequence is also $x$-promoting, so when this event
happens it also happens that $\labeling_{01},\labeling_{10},
\labeling_{11}$ are also made up entirely of $x$-promoting
blocks.  Furthermore, by symmetry, all labelings in the
set
$\{\labeling_{00},\labeling_{01},\labeling_{10},\labeling_{11}\}$
are equiprobable given this event.  If we can show that
at least one of these four labelings has an 
$x$-firewall incubator that starts in $B_1$ and
ends in $B_r$ then we will have shown that the 
probability of such an incubator existing is at
least $\frac14 c_0^{r+2}$, thus establishing the lemma.

Recall the bias of a node, $\beta(i)$, defined as the 
sum of the $2w+1$ signs associated to the nodes within
distance $w$ of $i$, including $i$ itself.  Note that
if $i$ belongs to the middle block in a sequence of 
three consecutive $x$-promoting blocks, then $\beta_0(i) > \sqrt{w}$.
In fact, letting $B,B',B''$ denote the three blocks and letting
$j$ denote the position of $i$ within $B'$, we have
$$
\beta_0(i) = \ccount(B') + \ccount_{-(j-1)}(B) + \ccount_j(B'') 
> 5 \sqrt{w} - 2 \sqrt{w} - 2 \sqrt{w}
$$
by the definition of an $x$-promoting block.
In particular, our assumption that all of the
blocks constituting the labelings
$\{\labeling_{00},\labeling_{01},\labeling_{10},\labeling_{11}\}$
are $x$-promoting implies that every node
in the blocks $B_1,\ldots,B_r$ has bias greater than $\sqrt{w}$ in
all four of the labelings.  

When $i$ belongs to $B_1$ or $B_2$, all of the nodes 
within distance $w$ of $i$ belong to $B_0 \cup B_1 \cup B_2 \cup B_3$.
Thus, when we reverse the ordering of labels of those $4w$ nodes,
the resulting sequence of biases in $B_1 \cup B_2$ is 
also reversed.  In particular, we can ensure that the set 
$M_1 = \arg \min \{ \beta_0(i) \mid i \in B_1 \cup B_2 \}$
intersects $B_1$ by retaining the labeling of $B_0,\ldots,B_3$
as in $\lambda_{00}$ or taking the reverse ordering.  Similarly, 
we can ensure that the set
$M_r = \arg \min \{ \beta_0(i) \mid i \in B_{r-1} \cup B_r \}$
intersects $B_r$ by retaining the labeling of $B_{r-2},\ldots,B_{r+1}$
or taking the reverse ordering.  It follows that in at least one of the labelings
$\{\labeling_{00},\labeling_{01},\labeling_{10},\labeling_{11}\}$,
the sets $M_1 \cap B_1$ and $M_r \cap B_r$ are both nonempty.
When this happens, by construction, any sequence of nodes
starting in $M_1 \cap B_1$ and ending in $M_r \cap B_r$ is 
an $x$-firewall incubator.
\end{proof}



\paragraph{The Lifecycle of an Incubator}  We focus on the interaction between the attacker and defender blocks on one side of a firewall incubator.  At every point in time, some swap of two nodes is proposed.  This swap contributes to the construction of a firewall in the defender if it involves an $o$ moving out of the defender, and hinders the construction if it involves an $x$ moving out of the attacker.  We need to show that good moves ($o$'s moving from the defender) happen sufficiently frequently and early in the process.  For the remainder of this section, we focus on moves in left attacker/defenders; similar statements hold for right attacker/defenders.  

More formally, we introduce a notion called {\it satisfaction time} which indicates the first time at which an element is selected for a move.
\begin{definition}
\label{def:sattm}
The {\emph{\sattm}} of a node $i$, denoted by
$\sat{i}$, is defined to be the first time when 
$i$ is selected to participate in a proposed swap 
with an unhappy, oppositely labeled individual. 
(If no such time exists, then $\sat{i}=\infty$.) 
A node $i$ is called \emph{impatient} at time $t$
if it is unhappy and $t \leq \sat{i}$.
\end{definition}
\noindent
Note that the element at $i$ may not actually participate in a swap at time $\sat{i}$, since it will only participate in a swap if it is unhappy.  In particular, an attacking $x$-element $i \in A_L$ may swap at its satisfaction time, swap at a later point, or not swap at any point.   However, a defending $o$-element $i\in D_L$ is guaranteed to swap at its satisfaction time if its bias and those of all its neighbors in the defender and internal blocks remain positive up until and including its own satisfaction time.  The initial bias of any such $i$ is, by the definition of an incubator, at least $\sqrt{w}$, and so the guarantee holds so long as enough attacking $x$-elements remain when the $o$-element's satisfaction time is reached.  To capture this intuition mathematically, 
we make the following definitions.
\begin{definition}
\label{def:combatants}
For a firewall incubator $F = D_L \cup I \cup D_R$ with corresponding attackers
$A_L,A_R$, a \emph{left attacking $x$} is an individual of type $x$ who belongs to $A_L$ in the initial configuration, and a \emph{left defending $o$} is an individual of type $o$ who belongs to $D_L$ in the initial configuration.  A \emph{left combatant} is an individual that is either a left attacking $x$ or a left defending $o$.  The equivalent terms with ``right'' in place of ``left'' are defined similarly; henceforth when referring to combatants we will omit ``left'' and ``right'' when they can be inferred from context.  The number of 
left attacking $x$'s and left defending $o$'s are denoted by $a_L,d_L$,
and for the right combatants we define $a_R,d_R$ similarly.
\end{definition}
\begin{definition}
The \emph{left-\sumseq} (resp. \emph{right-\sumseq}) 
is the sign sequence obtained by listing all of 
the left (resp. right) 
combatants in reverse order of satisfaction time, and translating
each attacking $x$ in this list to $+1$ and each defending $o$ to $-1$.

If there exists a time $t_0$ at which 
no individuals in $F$ are impatient,
any sign sequence obtained from the left-\sumseq (resp. right-\sumseq)
by permuting the signs
associated to individuals whose \sattm is after $t_0$, while fixing
all other signs in the \sumseq, is called a \emph{left-\psumseq}
(resp. \emph{right-\psumseq}).
\end{definition}
The relevance of \psumseq{}s will only become clear much later,
when we prove Proposition~\ref{prop:partialsums}.  They constitute
a relaxation of \sumseq{}s that encode almost all of the relevant
information in the \sumseq{}s --- since they differ from \sumseq{}s
only by a permutation that is, in some sense, irrelevant --- yet 
they turn out to be easier to work with probabilistically.

Define the $k^{\mathrm{th}}$ partial sum of a sequence to be the sum of its first $k$ elements. 
Our result follows from the following two main propositions.  
\begin{prop}
\label{prop:firewallforms}
Suppose that $F$ is a firewall incubator and
there exist left- and right-\psumseq{}s such that all partial sums 
of both \psumseq{}s are non-negative.  Then $F$ becomes an $x$-firewall.
\end{prop}
\begin{proof}
The proof is by contradiction.  Let $t_0$ denote the earliest time at which
no individual in $F$ is impatient; such a $t_0$ exists by our hypothesis that
left- and right-\psumseq{}s exist.
If $F$ is not an $x$-firewall at time $t_0$, then some node $j \in F$ contains an individual of type $o$ at that time.  There must exist a time $t_1 \leq t_0$ at which $j$ is occupied by a \emph{happy} individual of type $o$.  The proof is by case analysis: $j$ is not impatient at time $t_0$, so either it is happy or $\sat{j} < t_0$.  If $j$ is happy, set $t_1=t_0$.  Otherwise, if the type-$o$ individual occupying node $j$ at time $t_0$ has never moved, then set $t_1 = \sat{j}$; the type-$o$ individual occupying $j$ must have been happy at time $t_1$ or else he would have moved at that time.  Finally, if the type-$o$ individual occupying node $j$ at time $t_0$ is not the original occupant, then let $t_1$ denote the time immediately after he moved to location $j$; by the definition of the swap rule, this means $j$ was occupied by a happy individual of type $o$ at time $t_1$.

Consider the first node in $F$ to develop a negative $x$-bias, and let $t$ denote the time when this happens.  Note that $t \leq t_1$ since node $j$ must have negative $x$-bias at time $t_1$ in order for the occupying type-$o$ node to be happy.  Up until time $t$, the biases of all nodes in $F$ are positive, and so the only swaps in $F$ involve $o$-elements moving out.  Such swaps can not decrease the $x$-bias of nearby nodes, and since the $x$-bias of nodes in $I$ is completely determined by the labels of nodes in $D_L\cup I\cup D_R = F$, we conclude that  the first node to develop a negative $x$-bias is not in $I$, but rather must be in $D_L$ or $D_R$.  Without loss of generality, suppose that it is node $i \in D_L$.  From the definition of a firewall incubator, the initial $x$-bias of $i$, $\beta_0(i)$, was bounded below by the $x$-bias of the leftmost node in $D_L$.  The set of neighbors of the leftmost node in $D_L$ (including the node itself) is $A_L \cup D_L$, so the $x$-bias of the leftmost node can be expressed in terms of the number of attacking $x$'s and defending $o$'s, $a_L$ and $d_L$, by the formula 
$$
a_L - (w-a_L) + (w+1-d_L) - d_L = 2(a_L-d_L) + 1.
$$
Hence, $\beta_0(i) \geq 2(a_L-d_L)+1$.
Up until time $t$, swaps involving elements in $I$ do not decrease the $x$-bias of $i$; we will ignore such swaps in the remainder of the proof.  In $D_L$, a swap happens before time $t$ if and only if it involves an $o$-element.  In particular, up until time $t$, whenever the satisfaction time of an $o$-element in $D_L$ is reached, it swaps out, increasing the $x$-bias of $i$ by $2$.  In the $A_L$ block, whenever an $x$-element swaps out, it decreases the $x$-bias of $i$ by $2$.  This happens in one of three ways:
\begin{enumerate}
\item An attacking $x$ (i.e. an $x$-element that was present in the initial state) swaps out at its satisfaction time (and before $t$).  This contributes $-2$ to the $x$-bias of node $i$.
\item An attacking $x$ swaps out after its satisfaction time (but still before $t$).  Again this contributes a $-2$ to the $x$-bias of node $i$.  
\item  An $x$-element swaps into the attacker, becomes unhappy, and later swaps out (all before time $t$).  In this case, the element contributes $+2$ to the $x$-bias of $i$ when it swaps in and $-2$ when it swaps out, so the total contribution at time $t$ is $0$.
\end{enumerate}
Let $a^t_L$ be the number of attacking $x$'s in $A_L$ whose satisfaction time is before $t$. The above shows that the decrement to the $x$-bias of $i$ due to swaps of elements in $A_L$ is at most $2a^t_L$. Similarly define $d^t_L$ to be the number of defending $o$'s in $D_L$ whose satisfaction time is before $t$.  Then we have that the $x$-bias of $i$ at time $t$ satisfies:
\begin{align}
\nonumber
\beta_t(i) & \geq \beta_0(i) + 2 d^t_L - 2 a^t_L \\
\nonumber
& \geq 2(a_L-d_L)+1+2d^t_L-2a^t_L \\
& > 2 \cdot [(a_L - a^t_L) - (d_L - d^t_L)].
\label{eqn:xbiasbound}
\end{align}
Recall that $a_L$ is the number of attacking $x$'s  and similarly $d_L$ is the number of defending $o$'s.  Thus $a_L-a^t_L$ is the number of attacking $x$'s whose satisfaction time is greater than $t$, and similarly $d_L-d^t_L$ is the number of defending $o$'s whose satisfaction time is greater than $t$.  Thus, the right side of~\eqref{eqn:xbiasbound} is twice the $k^{\mathrm{th}}$ partial sum of the left-\sumseq, where 
$k = (a_L-a^t_L) + (d_L-d^t_L)$
denotes the number of individuals whose satisfaction time is after $t$.
As $t$ is earlier than $t_0$, the earliest time at which no individuals 
in $F$ are impatient, any left-\psumseq differs from the left-\sumseq only
by permuting a subset of the first $k$ signs, and therefore has the same
$k^{\mathrm{th}}$ partial sum.  Now our assumption that the $x$-bias of $i$
becomes negative at $t$ contradicts the hypothesis that there exists a
left-\psumseq whose partial sums are all non-negative.
\end{proof}

The proof of the next Proposition occurs in Appendix~\ref{sec:partialsums}.
\begin{prop}
\label{prop:partialsums}
If $B$ is a random block of length $6w$, then with probability
$\Omega(1/w)$, $B$ contains a firewall incubator having left- and 
right-\psumseq{}s
whose partial sums are all non-negative.
\end{prop}
The following gives the intuition behind the proof.  
First, we know from Proposition~\ref{prop:incubators-happen}
that with constant probability, $B$ contains a firewall
incubator $F$.  Let us focus on the 
left-\sumseq of $F$.  
First assume (unjustifiably) that the \sumseq is a uniformly 
random permutation of the $a_L$ $+1$'s and $d_L$ $-1$'s.  Then, by the 
following simple probabilistic lemma known as the {Ballot Theorem}, 
its partial sums are all non-negative with probability 
$(a_L-d_L)/(a_L+d_L)$ which, by the definition of a
firewall incubator, is at least $\Omega(\sqrt{1/w})$.

\begin{lem}[Ballot Theorem] \label{lem:cyclic}
Consider a multiset of consisting of $a$ copies
of $+1$ and $b$ copies of $-1$, and let 
$x_1,x_2,\ldots,x_{a+b}$ be a uniformly random
ordering of the elements of this multiset.
The probability that all partial sums
$x_1+\cdots+x_j \; (1 \leq j \leq a+b)$ are 
strictly positive is equal to $\max\{0,\tfrac{a-b}{a+b}\}$.
\end{lem}
The theorem was first proved in 1887~\cite{Ballot1,Ballot2,Renault}.
One elegant proof, originally due to 
Dvoretzky and Motzkin~\cite{DvorMotz}, 
is presented in Appendix~\ref{sec:ballot}.

Unfortunately, the \sumseq is not a uniformly random permutation.  A bias arises since the number of unhappy elements of each type is not precisely equal.  If at some point there are more unhappy $o$'s, say, than $x$'s, then the satisfaction time of an attacking $x$ is more likely to happen earlier.  In Section~\ref{subsec:unhappies}, we will show that the number of unhappy elements is approximately balanced for a sufficiently long time.  Now at any point in time, we artificially correct the small imbalance as follows: suppose there are $m$ extra unhappy elements of one type, say $x$.  Then choose $m$ unhappy $x$-elements at random and call them {\it censored}.  Call a swap {\it censored} if it involves a censored element.  (There is actually one more technicality here: we also want to censor swaps if both elements are combatants of $F$.  This necessitates a subtle modification to the censorship construction, with no significant quantitative consequences for the proof.)  Since swaps are between random unhappy elements, as long as the imbalance is small, the probability that a swap is censored is also small.  Conditioning on having no censored swaps, the \sumseq is indeed a uniformly random permutation.  There is a stopping time $t_0$ at which the imbalance ceases to be small, and we cannot guarantee that censored swaps are unlikely after $t_0$; however, we can show that with probability $1-o(1)$, no individual in $F$ is impatient at time $t_0$.  Consequently, we can obtain a \psumseq by randomly permuting the combatants whose satisfaction times are after $t_0$, and provided that no censored swaps occurred before $t_0$, the \psumseq is a uniformly random permutation.  Then, Proposition~\ref{prop:partialsums} follows from Lemma~\ref{lem:cyclic}.

Our main theorem follows fairly directly from Proposition~\ref{prop:firewallforms} and Proposition~\ref{prop:partialsums}.

\begin{theorem}
Consider the segregation process with window size $w$ 
on a ring network of size $n$, starting from a uniformly random
initial configuration.  There exists a constant $c < 1$ and a 
function $n_0 : \mathbb{N} \to \mathbb{N}$ such that for all
$w$ and all $n \geq n_0(w)$, with probability $1-o(1)$, the
process reaches a configuration 
after finitely many steps in which no further swaps are 
possible.  The average run length in this final configuration
is $O(w^2)$.  In fact, the distribution of runlengths in the 
final configuration
is such that for all $\lambda>0$, the probability of a randomly selected node
belonging to a run of length greater than $\lambda w^2$ is 
bounded above by $c^{\lambda}$.
\end{theorem}
\begin{proof}
By Proposition ~\ref{prop:frozenstate}, with high probability, 
the process reaches a 
frozen configuration in which no further swaps are possible.
To bound the distribution of runlengths in the frozen configuration, 
consider a randomly sampled site, once
again denoted by $a$.  As we scan clockwise from $a$ in the
initial configuration, let us divide the ring into disjoint
blocks of length $6w$.  Each of these blocks has
probability $\Omega(1/w)$ of containing an $x$-firewall
incubator
having left- and right-pseudo-transcripts whose partial sums 
are all non-negative.  (Proposition~\ref{prop:partialsums}.)
Of course, similar statements hold
with $o$ in place of $x$.  Thus,
for a suitable constant $c < 1$,
the probability that none of the 
first $\lambda w / 6$ blocks encountered
on a clockwise scan of length-$(6w)$ blocks
starting from $a$ contain $x$-firewalls in the final 
frozen configuration is bounded above by $c^{\lambda}$.
The same conclusion holds with $o$ in place of $x$
and with counterclockwise in place of clockwise, by
symmetry.  Node $a$ cannot belong to a monochromatic
run of length greater than $\lambda w^2$ 
assuming that it has individuals of both labels within this
radius on both sides of itself, and this completes the proof.
\end{proof}

%% file: unhappies-summ.tex
\newcommand{\sv}{{\zeta}}
\newcommand{\esv}{{\zeta}}
\newcommand{\vctr}[1]{{\bm{#1}}}
\newcommand{\Z}{{\mathbb{Z}}}
\newcommand{\xx}{{x}}
\newcommand{\oo}{{o}}
\newcommand{\allseq}{{\{\xx,\oo\}^{1 \ldots n}}}
\newcommand{\stoptime}{T_0}
\newcommand{\taint}{{D}}
\newcommand{\ntaint}{{d}}
\newcommand{\invol}{\iota}

This section sketches a proof 
that the numbers
of unhappy $\xx$'s and $\oo$'s remain nearly balanced
until late in the segregation process.  More precisely,
define a stopping time $\stoptime$ to be the earliest time
when fewer than $3n/w^2$ individuals are impatient.
Theorem~\ref{thm:balance} below asserts that when
the ring size $n$ is sufficiently large, it holds
with high probability that at all times $t \leq \stoptime$
the numbers of unhappy $\xx$'s
and $\oo$'s differ by at most $n/w^4$.
The full proof of the theorem is given in Appendix~\ref{sec:wormald}.
The theorem statement is plausible because 
the entire stochastic process is symmetric under interchanging
the roles of $\xx$ and $\oo$.  Thus, one would expect nearly
equal numbers of unhappy $\xx$'s and $\oo$'s in the initial
configuration, and one would expect this near-balance to persist
for many steps after the initialization of the 
stochastic process, since the process itself has no
bias in favor of reducing the number of unhappy $\xx$'s 
more rapidly than the number of unhappy $\oo$'s or vice-versa.

We can express the segregation process as a continuous-time
process (with state changes only at times that are multiples of 
$1/n$) whose state variables encode, for every 
$k$ and every sequence $\sigma = (\sigma_0,\ldots,\sigma_k) \in 
\{\xx,\oo\}^k$, the fraction of sites such that $\sigma$
describes the labeling of that site and its $k$ nearest
clockwise neighbors.  Any other parameter depending only
on the frequency of occurrence of certain bounded-size 
configurations (e.g.~the fraction of unhappy $\xx$'s
and $\oo$'s) can be expressed as a function of these state
variables.  
In the continuum limit, the state variables
are not random at all; they evolve deterministically
according to a system of differential equations.
The operation of interchanging the symbols $\xx$ and $\oo$ defines
a permutation of the set of state vectors, and the fixed-point
set of this permutation is an invariant set for the differential equation
because the derivative at any such point must also (by symmetry)
be preserved under the operation of interchanging $\xx$
and $\oo$ and therefore the differential equation solution
can never exit the set of vectors preserved by this operation.
In the
continuum-limit process the initial state belongs to this
fixed-point set and therefore the state vector at all future
times remains invariant under interchanging $\xx$ and $\oo$,
which gives a heuristic justification of the fact that the 
numbers of unhappy $\xx$'s and $\oo$'s remain nearly equal
throughout the segregation process, or until the process becomes so close to ``frozen" that the continuum-limit approximation no longer applies.

Wormald's technique~\cite{wormald} provides a mathematically
rigorous method for justifying these differential-equation
approximations and quantifying the approximation error.  
Using this technique, 
we derive the associated differential equation as a limit of discrete-time difference equations and show that it is quadratic with coefficients independent of $n$. However, our problem resists a straightforward application of Wormald's
technique because it has infinitely many state variables with
infinitely long dependency chains.\footnote{Wormald's technique is
applicable to differential equations with state variables
$Y_1,Y_2,\ldots$ such that the derivative of $Y_i$ depends 
only on $Y_1,\ldots,Y_{i}$ for all $i$, but unfortunately
our variables do not admit such an ordering.}  

To circumvent
this difficulty, we do not directly analyze the segregation 
process on a ring of size $n$.  Instead, for a function $L(w)$
that grows sufficiently rapidly, we analyze the
segregation process on a disjoint union of $n/L(w)$ rings 
each having length $L(w)$.\footnote{For simplicity, we 
assume that $n$ is divisible by $L(w)$.  In general, we 
would have to analyze a disjoint union of $\lfloor n/L(w) \rfloor$
rings each having length $L(w)$ or $L(w)+1$.}  
We then bound the error resulting from approximating a single ring by a collection of bounded-length
rings by defining a notion of ``taint'' that allows us to couple the
single-ring and bounded-length ring versions of the process.  Taint is defined in such a way that for any untainted node, the node and its neighbors have the same label in both versions of the process. We then bound the number of untainted sites with high probability using martingale techniques.   Combining these arguments, we obtain the following result:
\begin{theorem} \label{thm:balance}
Consider the segregation process on a ring of length $n$ with
window size $w$.
For all $w$ and all sufficiently
large $n$ (i.e.\ all $n \geq n_0(w)$, for 
some function $n_0$),
with probability greater than 
$1 - \tfrac{1}{w}$,
the number of unhappy $\xx$'s differs from the
number of unhappy $\oo$'s by at most $n/w^4$ at
every time $t \leq T_0$, where $T_0$ denotes the
earliest time when 
fewer than $3n/w^2$ individuals are impatient.
\end{theorem}
As mentioned earlier, the full proof is given in Appendix~\ref{sec:wormald}.

%% file: openq.tex
\label{sec:openq}

In this paper we have shown that the one-dimensional
Schelling segregation process with window size $w$
leads, with high probability, to a ``frozen configuration''
in which most nodes belong to monochromatic
runs of size at least $\Omega(w)$ and at most $O(w^2)$.
We are hopeful
that the upper bound can be improved to $O(w)$, using 
an extension of the techniques introduced 
in this paper, but this strengthening of the result
is beyond the scope of the present work.  Assuming
it is correct that most nodes belong to monochromatic
runs of size $\Theta(w)$, it becomes natural to
conjecture that the distribution of runlengths, normalized
by $1/w$, converges to a distribution $F$.  In other
words, if
$n=n(w)$ grows sufficiently fast as a function of $w$,
then for all $r > 0$,
as $w \to \infty$ 
the probability that a randomly selected node belongs to a 
run of length less than $rw$ in the frozen configuration
converges to a limit $F(r)$.  Proving such a limit
theorem seems beyond the reach of the techniques
introduced here, to say nothing of characterizing
the precise runlength distribution $F$, if it exists.

Several parts of our analysis hinged on symmetry arguments
that are specific to the case in which $x$ and $o$
are equally likely in the initial configuration, and in
which the threshold 
$\tau$, defining the fraction of neighbors that must
be of the same type as an individual in order for that
individual to be satisfied, is equal to $\tfrac12$.  It is 
quite possible that when one varies either of these 
assumptions, the model's behavior is qualitatively 
different; for example, the lengths of the runs in the
frozen configuration may become exponential rather than
polynomial in $w$.  Understanding how the 
one-dimensional model's behavior
varies as we vary the parameter $\tau$ or the 
$x$-to-$o$ ratio are important questions for
future work.  A related open problem is to analyze
the one-dimensional segregation process when the tolerance 
threshold $\tau$ may vary from one individual to another.

Finally, and most ambitiously, there is the open
problem of rigorously analyzing the Schelling model
in other graph structures including two-dimensional
grids.  Simulations of the Schelling model in two
dimensions reveal beautiful and intricate patterns
that are not well understood analytically.  
Perturbations of the model have been successfully 
analyzed using stochastic stability analysis
by Zhang~\cite{zhang_dynamic_2004,zhang_residential_2004,zhang_tipping_2011},
but the non-perturbed model has not been rigorously
analyzed.  Two-dimensional lattice models are almost
always much more challenging than one-dimensional ones,
and we suspect that to be the case with Schelling's
segregation model.  But it is a challenge worth undertaking:
if one is to use the Schelling model to gain insight into
the phenomenon of residential segregation, it is vital to
understand its behavior on two-dimensional grids
since they reflect the structure of so many 
residential neighborhoods  in reality.

%% file: ack.tex
We are grateful to Jon Kleinberg for introducing us to the fascinating
topic of Schelling's segregation model and to Alistair Sinclair for
supplying us with valuable insights in the early stages of this work.

%% file: runapp.tex
\subsection{Proof of Proposition~\ref{prop:partialsums}}
\label{sec:partialsums}


\newcommand{\INC}{{\mbox{\sc inc}}}
\newcommand{\NN}{{\mbox{\sc non-neg}}}
\newcommand{\A}{{\mbox{\sc swap}}}
\newcommand{\B}{{\mbox{\sc time}}}
\newcommand{\Pad}{{\mbox{\sc pad}}}
\newcommand{\ts}{{\tau}}


We would like to show that a random block evolves into a firewall.   We already know that such a block contains a firewall incubator with constant probability, and that an incubator becomes a firewall if there are pseudo-transcripts such that all partial sums are non-negative.  Thus the crux of the argument is to show that such pseudo-transcripts exist with sufficiently high probability.  This would follow from Lemma~\ref{lem:cyclic} if the transcripts were random permutations, but that is not precisely true.  The reason is that the global number of unhappy elements of each type might be imbalanced, creating an imbalance in the probability of an $o$-swap versus an $x$-swap.  We correct this imbalance by censoring certain swaps and then conditioning on the event that the transcript involves no censored swaps.

We begin with a definition of {\it censored swaps}.  
%
%
Fix a block $B$ 
%
%
and a time $t$.  Let the padded block $\Pad(B)$ consist of $B$ together with the $w$ nodes on its left and the $w$ nodes on its right. Let $n_x(t,B)$ be the number of unhappy $x$-elements {\it outside} $\Pad(B)$ at time $t$, and define $n_o(t,B)$ similarly.  Suppose that $n_x(t,B)\leq n_o(t,B)$ (the other case is similar), and let $k=n_o(t,B)-n_x(t,B)$.  Then the censored pairs involving an $o$-element in $\Pad(B)$ are precisely those involving an $x$-element in $\Pad(B)$.  Let $C$ consist of an arbitrary subset of $k$ unhappy $o$-elements outside $\Pad(B)$.  For an $x$-element inside $\Pad(B)$, the censored pairs are those involving an $o$-element in $\Pad(B)\cup C$.  All other pairs are uncensored.  A proposed swap is censored if the corresponding pair is censored.  

The following properties of this definition will be useful in our proof.  The first two properties will help us prove that transcripts of uncensored swaps are uniformly random permutations.  The third property will be used to show that censored swaps are rare.  For the third property to follow, we need to argue that there are sufficiently many unhappy elements so that the ratio of censored to uncensored swaps is small.  
To this end, recall from Theorem~\ref{thm:balance} that up until time $T_0$, the combined number of unhappy individuals is at least $3n/w^2$.  

\begin{lem}
\label{lem:censorprop}
For any block $B$ and time $t$,
\begin{enumerate}
\item \label{cprop:1}
every pair of elements in $\Pad(B)$ is censored,
\item \label{cprop:2}
every element in $\Pad(B)$ is in an equal number of uncensored pairs,
\item \label{cprop:3}
and if $t<T_0$, then with probability $1-1/w$, every unhappy element in $\Pad(B)$ has at most $n/w^4+|\Pad(B)|$ censored and at least 
$\tfrac32 n/w^2-n/w^4-|\Pad(B)|$ uncensored partners.
\end{enumerate}
\end{lem}
\begin{proof}
The first two properties follow by construction.  The third property follows immediately from Theorem~\ref{thm:balance}.
\end{proof}

We are now ready to prove Proposition~\ref{prop:partialsums} which we restate here for convenience.

\begin{prop}
If $B$ is a random block of length $6w$, then with probability
$\Omega(1/w)$, $B$ contains a firewall incubator having left- and 
right-\psumseq{}s
whose partial sums are all non-negative.
\end{prop}


\begin{proof}
Let $\INC$ denote the event that the initial labeling of
$B$ contains a firewall incubator.  
We know from Proposition~\ref{prop:incubators-happen} that
$\Pr(\INC)=\Omega(1)$.  Let us fix a firewall incubator $F$ in $B$
(if it exists) for the remainder of the proof.  
Consider the \sumseq $\ts$ that combines the left- and right-\sumseq{}s
of $F$, i.e., the sign sequence obtained by listing all of the 
combatants of $F$
(both right and left) 
in reverse order of satisfaction time and translating each combatant 
in the list to the appropriate sign.  
Let $t_0$ be the earliest time, if it exists, at which no individuals
in $F$ are impatient.  (If no such time exists, $t_0 = \infty$.)
Define the \emph{scrambled pseudo-transcript}, $S(\ts)$, to be a 
sign sequence obtained from $\ts$ by permuting, uniformly at random, the signs
associated to nodes whose satisfaction time is after $t_0$.
Let $S(\ts_L), \, S(\ts_R)$ be the subsequences of $S(\ts)$
corresponding to left- and right-combatants.  Our goal is to show
that, with probability $\Omega(1/w)$, $S(\ts_L)$ and $S(\ts_R)$
are left- and right-\psumseq{}s for $F$ whose partial sums are
all non-negative.  Denote this event by $\NN$; note that $\NN$
is a subset of $\INC$, since otherwise the incubator $F$ and 
its combined transcript $\ts$ are not even defined.

We define two more events to be used throughout this proof.  
Let $\Pad(F)$ denote the block consisting of $F$ and its 
associated attacker blocks $A_L,A_R$.  Recall that $t_0$
is the earliest time at which no individuals in $F$ are
impatient, or $t_0=\infty$ if no such time exists.
Then the two events of interest are:
\begin{itemize}
\item $\A$, the event that no unhappy element of $\Pad(B)$ participates in a censored swap before time $t_0$,
\item $\B$, the event that $t_0<T_0$.  (Recall that $T_0$ is the earliest time
when fewer than $3n/w^2$ individuals are impatient.)
\end{itemize}
We would like to prove $\Pr[\NN]=\Omega(1/w)$.
We do this 
using the following pair of claims.
\begin{claim} $\Pr[ \NN | \A, \INC]=\Omega(1/w)$.\end{claim}
\begin{proof}
Let $i_1,\ldots,i_\ell$ denote the list of combatants in the order 
that their signs appear in $S(\ts)$.  In other words, the sequence 
$i_1,\ldots,i_\ell$ consists of a random permutation of the
combatants whose satisfaction time is later than $t_0$, 
followed by a listing of the others in decreasing order
of satisfaction time.  
We first show that for any initial labeling of $\Pad(F)$, 
given event $\A$, the list $i_1,\ldots,i_\ell$ 
is a uniformly random permutation of the combatants.  
First note that, by property~\ref{cprop:1} of Lemma~\ref{lem:censorprop}, 
combatants only swap with non-combatants, so the list is unique 
(i.e., there are no two combatants with equal satisfaction times).  
To prove that the list is a random permutation, we show that
for every $k$, if we condition on the subsequence
$i_{k+1},\ldots,i_\ell$ (in addition to conditioning on event $\A$), 
then $i_k$ is a uniformly random sample from the set $C_k$ of combatants
not listed in $i_{k+1},\ldots,i_\ell$.
This is proven by analyzing two cases.
Denote the satisfaction time of $i_k$ by $t^*$.
If $t^* > t_0$, then by construction $i_1,\ldots,i_k$ is 
a uniformly random permutation of $C_k$, so $i_k$ is
a uniformly random sample from $C_k$.  On the other 
hand, if $t^* \leq t_0$, 
then by property~\ref{cprop:2} of Lemma~\ref{lem:censorprop},
each of the individuals in $\{i_1,\ldots,i_k\}$ is equally
likely to participate in an uncensored swap at time $t^*$.


Note that this means $S(\ts_L)$ and $S(\ts_R)$ are also independent uniformly random permutations for any initial configuration of $\Pad(F)$ given events $\A,\INC$.  As $\beta_0(i)>\sqrt{w}$ for all $i\in F$, and in particular for the left-most and right-most nodes, we know that $a_L-d_L > \tfrac12 \sqrt{w}$
(and similarly for $a_R-d_R$).  Furthermore, $a_L+d_L\leq 2w+1$ (and similarly for $a_R+d_R$).  Thus, applying the Ballot Theorem (Lemma~\ref{lem:cyclic}), we see that the probability $S(\ts_L)$ (similarly $S(\ts_R)$) has non-negative partial sums given events $\A$ and $\INC$ is $\Omega(1/\sqrt{w})$.  Therefore, by the independence of $S(\ts_L)$ and $S(\ts_R)$, $\Pr[ \NN | \A,\INC]=\Omega(1/w)$ as claimed.
\end{proof}

\begin{claim}$\Pr[\A, \INC] = \Omega(1)$ for $w$ sufficiently large.
\end{claim}  
\begin{proof}
Recall the event $\B$, that $t_0 < T_0$.
We will establish the stronger claim that $\Pr[\A,\B,\INC]=\Omega(1)$
for $w$ sufficiently large.

We first work on bounding $\Pr[\A | \B, \INC]$.
Given event $\B$, we know from Theorem~\ref{thm:balance}
that there is at most $O(1/w)$ probability of
there being a time $t \leq t_0$ at which the
numbers of unhappy $\xx$'s and $\oo$'s differ by more than $n/w^4$.
If there is no such time $t$ then
property~\ref{cprop:3} of Lemma~\ref{lem:censorprop} implies that
the probability the proposed swap for a combatant is censored is 
at most $O(1/w^2)$.  By the union bound, the probability that 
any proposed swap in the transcript 
is censored is at most $O(1/w)$.  
Thus,
$\Pr[\A | \B, \INC] = 1-O(1/w)$.

It remains for us to show that $\Pr[\B,\INC] = \Omega(1)$.
We will use the equation
\begin{equation} \label{eq:binc}
\Pr[\B,\INC] = \Pr[\INC] - \Pr[\INC \wedge \overline{\B}].
\end{equation}
The first term on the right side is $\Omega(1)$,
by Proposition~\ref{prop:incubators-happen}.
To bound the second term, recall that
by definition of $T_0$, 
the total number of impatient individuals 
in the ring is less than $3n/w^2$.
Since $B$ is a random block of size $6w$,
the expected number of impatient individuals
in $B$ at time $T_0$ is less than $18/w$, so
by Markov's inequality, the (unconditional) probability that
our random block $B$ contains any impatient
individual at time $T_0$ is at most $18/w$.
However, if event $\INC \wedge \overline{\B}$
occurs, it means that $F$ must contain an 
impatient individual at
time $T_0$.  Thus, $\Pr[\INC \wedge \overline{\B}] \leq 18/w$.
Equation~\ref{eq:binc} now says that
$\Pr[\B,\INC] = \Omega(1) - O(1/w) = \Omega(1)$,
for large enough $w$.
\end{proof}
\noindent Noting that $\Pr[\NN] \geq \Pr[\NN \mid \A,\INC] \cdot \Pr[\A,\INC]$
and applying the preceding two claims finishes the proof. \end{proof}

\subsection{Partial Sums of Random Permutations}
\label{sec:ballot}

This subsection proves Lemma~\ref{lem:cyclic}, the 
Ballot Theorem.  The first proof appeared in 1887,
and many proofs have been discovered since then.
We present one such proof, which is essentially the 
one given by Dvoretzky and Motzkin~\cite{DvorMotz},
in order to make our exposition more self-contained.

\begin{proof}
We will prove a stronger assertion: if a fixed
sequence $y_1,y_2,\ldots,y_{a+b}$ containing $a$ copies
of $+1$ and $b$ copies of $-1$ is permuted by a random
cyclic permutation to obtain $x_1,\ldots,x_{a+b}$, 
then the probability that all partial sums 
$x_1+\cdots+x_j \; (1 \leq j \leq a+b)$ are 
strictly positive is $\max\{0,\tfrac{a-b}{a+b}\}$.
This suffices to prove the lemma, since a uniformly
random permutation composed with a random cyclic
permutation is again uniformly random.

If $a \leq b$ then the partial sum 
$x_1+\cdots+x_{a+b}$ is non-positive, which implies
that the probability is zero, as claimed.  Henceforth
assume $a>b$.  Extend the sequence $x_1,x_2,\ldots,x_{a+b}$
to a period infinite sequence $x_1,x_2,\ldots$ with period
$a+b$, and let $s_j$ be the $j^{\mathrm{th}}$ partial sum
of this sequence, $s_j = \sum_{i=1}^j x_i$.  (When $j=0$,
we define $s_j$ to be zero.)  The identity
$s_{j+a+b} = s_j + a - b$ implies that $s_j \to \infty$
as $j \to \infty$.  Hence, for any integer $z$ there are at 
most finitely many $j$ such that $s_j=z$; define $t(z)$ to be
the largest such $j$, unless there is no $j$ with $s_j=z$
in which case $t(z)$ is undefined.  We know at least that
$t(z)$ is defined for all $z \geq 0$, because $s_0=0$ and 
$s_{j+1}$ exceeds $s_j$ by at most 1, for all $j$. 

Figure~\ref{fig:perm} illustrates the extension of a sequence and shows for several choices of $z$ the corresponding value $t(z)$.
\begin{figure}[t]
  \centering
\includegraphics[width=4.5in]{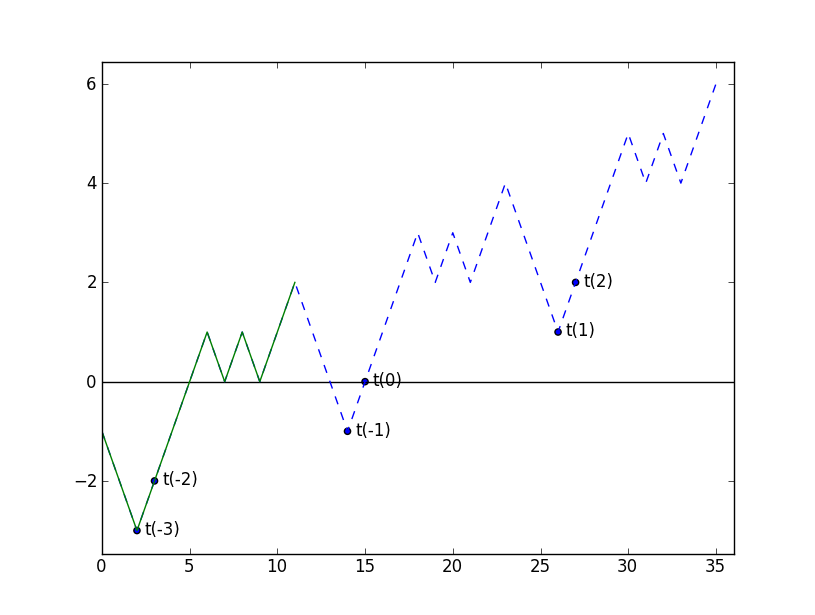}
  \caption{The partial sums of the extended sequence.}
\label{fig:perm}
\end{figure}

Once again using the identity $s_{j+a+b} = s_j + a - b$,
we find that $t(z+a-b) = t(z) + a - b$ for all $z$ such
that $t(z)$ is defined.  Thus, the image of $t$ (i.e.,~the
set of all $j$ such that $t(z)=j$ for some $z$) consists
of all the non-negative integers in exactly $a-b$ distinct
congruence classes  modulo $a+b$.  
Note
that 0 belongs to the image of $t$ if and only if
$t(0)=0$, and that this happens if and only if 
all of the partial sums $s_j \; (1 \leq j \leq a+b)$ are
strictly positive.  Thus, we devote the rest of the proof
to showing that $\Pr(0 \mbox{ belongs to the image of } t)
= \tfrac{a-b}{a+b}$.

Let us consider how the image of $t$ changes when we 
cyclically permute the sequence $x_1,x_2,\ldots,x_{a+b}$.
Denote the permuted sequence by 
$x'_1,x'_2,\ldots,x'_{a+b} = x_2,x_3,\ldots,x_1.$
As before make $x'_1,x'_2,\ldots$ into an infinite periodic
sequence and
denote its partial sums by 
$$
s'_j = \sum_{i=1}^j x'_i = \sum_{i=2}^{j+1} x_i = s_{j+1} - x_1.
$$
Let $t'(z)$ be the largest $j$ such that $s'_j = z$.  From
the formula $s'_j = s_{j+1}-x_1$ we immediately see that 
$t'(z) = t(z+x_1) -1$ for all $z$ such that both sides are 
defined.  Thus, the congruence classes in the image of $t'$
are obtained from those in the image of $t$ by subtracting 1
modulo $a+b$.  As we run through all of the cyclic permutations
of the sequence $x_1,x_2,\ldots,x_{a+b}$, we shift the set of 
congruence classes in the image of $t$ by subtracting each
element of $\mathbb{Z}/(a+b)$.  Since the image of $t$ contains
exactly $a-b$ congruence classes, the probability that the 
image contains $0$ when we apply a random cyclic permutation
is exactly $\tfrac{a-b}{a+b}$.
\end{proof}

%% file: wormald.tex
This section provides formal proofs of the arguments in Section \ref{subsec:unhappies}.

\paragraph{Tainted nodes and bounded-length ring approximation}
Assume now that $n$ is divisible by $L=L(w)$, and consider
a modified segregation process in which the nodes are still
numbered $1,2,\ldots,n$, but the graph structure $G$ is now a 
union of disjoint cycles of length $L$.  Specifically, nodes
$i$ and $j$ are connected if and only if 
$\lfloor i/L \rfloor = \lfloor j/L \rfloor$
and $i \equiv j \pm 1 \pmod{L}$.  Note that for large $L$,
the overwhelming majority of nodes
have the same set of neighbors in $G$ as in the
$n$-cycle; the exceptions are those $i$ that belong
to one of the congruence classes $-w+1, -w+2, \ldots, w-2, w-1$
mod $L$.

We will couple the segregation process in $G$
with the segregation process in the $n$-cycle $C_n$
in the obvious way: every time we choose two
locations $i,j$ for a proposed swap in $C_n$, we choose
the same pair of locations $i,j$ for a proposed swap in $G$.
Starting from identical labelings at time 0, differences 
in the two labelings will emerge whenever the coupling proposes
a swap between two locations that are both unhappy in one graph
but not the other.  To bound the rate at which
such differences develop, we define the set of \emph{tainted}
nodes at time $t$, $\taint(t)$, by the following inductive definition.
At $t=0$ the set of tainted nodes consists of all the nodes
$i$ belonging to the congruence classes $-w+1, -w+2, \ldots, w-2, w-1$
mod $L$.  If the two nodes $i,j$ that are chosen for a proposed swap
at time $t>0$ are both untainted, then $\taint(t) = \taint(t-1)$.
Otherwise $\taint(t)$ is the union of $\taint(t-1)$ with the set
of all nodes that are $w$-step neighbors of either $i$ or $j$ 
in either $G$ or $C_n$.  Note that 
\begin{equation} \label{eq:taint-6w}
|\taint(t)-\taint(t-1)| \leq 6w
\end{equation}
since the number of $w$-step neighbors of any one node
is $2w$ in each graph, and at least $w$ of them are $w$-step
neighbors in both graphs.
\begin{lem} \label{lem:untainted}
At any time $t$ in the coupling of the segregation processes on 
$G$ and $C_n$, if node $i$ is untainted, then $i$ and all of its
$w$-step neighbors have the same label in both graphs.
\end{lem}
\begin{proof}
The proof is an easy induction on $t$.  When $t=0$ this follows
from the fact that both graphs have the same initial labeling, and
a node has the same $w$-step neighbor set in both graphs unless
it belongs to $\taint(0)$.  When $t>0$, the only nodes whose 
$w$-step neighborhood may experience a label change are those 
located within $w$ steps of the nodes $i,j$ that were selected
for the proposed swap.  If $i$ or $j$ is tainted at time $t-1$
then all such nodes become tainted at time $t$.  
If both $i$ and $j$ are untainted at
time $t-1$, then the induction hypothesis implies that the proposed
swap affects the labeling of both graphs in the same way.
\end{proof}

\begin{lem} \label{lem:tainted}
The expected number of tainted nodes at time $t \geq 0$ is bounded
above by $e^{12wt/n} \left( \tfrac{2w-1}{L(w)} \right) n$.  The probability
that $\tfrac{|\taint(t)|}{n} > 2 e^{12wt/n} \left( \tfrac{2w-1}{L(w)} \right)$ 
is at most $\exp \left( - \tfrac{wn}{36 L(w)^2} \right)$.
\end{lem}
\begin{proof}
Let $i,j$ be the pair of locations that are chosen for a proposed
swap at time $t$.  Let $\ntaint(t) = |\taint(t)|$.
By the union bound, the probability that either
$i$ or $j$ belongs to $\taint(t)$ is at most $2\ntaint(t)/n$.  
If so, the $w$-step neighbors of $i$ and $j$ in $G$ and $C_n$ are added to
$D(n+1)$.  
We have seen in Equation~\eqref{eq:taint-6w} above
that $\ntaint(t+1) - \ntaint(t) \leq 6w$ in that case, and otherwise 
$\ntaint(t+1) = \ntaint(t)$.  
Therefore, 
$$
\expect{\ntaint(t+1) \mid \ntaint(t)} 
\leq \ntaint(t) + 6w \cdot (2 \ntaint(t) / n)
= \ntaint(t) \cdot \left( 1 + \frac{12w}{n} \right)
< \ntaint(t) \cdot e^{12w/n}.
$$
This shows that the sequence of random variables 
$Y_t = \ntaint(t) e^{-12wt/n}$ is a supermartingale.
The first statement of the lemma follows immediately
from this fact, together with the initial condition 
$Y_0 = \left(\tfrac{2w-1}{L(w)}\right)n$.  We prove the second statement
using Azuma's Inequality, which means we first need an upper bound
of the form $|Y_t - Y_{t+1}| \leq c_t$ almost surely.  
The calculation
\begin{align*}
|Y_t - e^{-12w/n} Y_t| + |Y_{t+1} - e^{-12w/n} Y_t| & \leq 12w (Y_t/n) + 6w e^{-12w(t+1)/n} < 18w e^{-12wt/n}
\end{align*}
justifies setting $c_t = 18w e^{-12wt/n}$.  Note that
$\sum_{t=0}^{\infty} c_t^2 = (18w)^2 (1 - e^{-24w/n})^{-1} < 18wn$,
from which the second statement of the lemma follows.
\end{proof}

\paragraph{Bounding the asymmetry of the process}
The segregation process on a disjoint union of rings of length $L(w)$
can be analyzed using Wormald's differential equation technique, a
method that applies to families of vector-valued 
discrete-time stochastic processes indexed by an integer $n$, 
whose behavior converges to a solution of a 
continuous-time differential equation as $n$ tends to
infinity.  We can representing a state of the segregation process
on $G$ by a $2^{L(w)}$ dimensional vector $\vctr{\sv}$
whose components are 
indexed by strings $\sigma \in \{\xx,\oo\}^{L(w)}$. 
The component $\sv_\sigma$ is defined to be the number
of nodes $i$ such that the ring containing $i$ 
is labeled
(say, in clockwise order starting from $i$) with the 
string $\sigma$. 


It is not difficult to work out a formula for the
conditional expectation $\expect{\esv_\sigma(t+1) \,|\,
\vctr{\esv}(t)}$.  
We simply have to enumerate all of the
ways that the number of rings with 
labeling $\sigma$ 
could increase or decrease in a single step, and work
out their probabilities.  For any $j \in \{1,\ldots,L(w)\}$
and $\sigma,\sigma',\sigma'' \in \{\xx,\oo\}^{L(w)}$,
consider a proposed swap between the $j^{\mathrm{th}}$ node of
a ring $R_1$ whose 
labeling (starting from the first node)
is $\sigma'$, and the first node of a ring $R_2$ whose 
labeling is $\sigma''$.  Define a coefficient
$a=a(j,\sigma,\sigma',\sigma'')$ as follows.  If 
$\sigma \neq \sigma'$ and the proposed swap results
in $R_1$ having 
labeling $\sigma$ instead of $\sigma'$, then $a = 1$.
If $\sigma=\sigma'$ and the proposed swap results
in $R_1$ having an 
labeling other than $\sigma$, then $a=-1$.  Otherwise $a=0$.
Having made these definitions, we can derive
the formula
\begin{eqnarray} 
\expect{\esv_\sigma(t+1) - \esv_\sigma(t) 
\,|\, \vctr{\esv}(t)=\esv}
 = 2 \sum_{j=1}^{L(w)} \! \sum_{\sigma',\sigma''}
a(j,\sigma,\sigma',\sigma'') \frac{\esv_{\sigma'} \esv_{\sigma''}}{n^2} 
 \, +  \, O \! \left( \! \frac{L(w)}{n} \! \right) \! .
  \label{eq:diffeq-discrete}
\end{eqnarray}
The first term accounts for the expected net change in the number of 
instances of $\sigma$ when going from step $t$ to step $t+1$.
 (The factor of 2
is because either of the two locations selected for the proposed 
swap at time $t$ could be the one that belongs to the ring $R_1$.)
None of the terms in the sum account for the probability 
that both selected locations belong to the same ring; 
this event has probability
$L(w)/n$, and it accounts for the $O(L(w)/n)$ correction term
at the end of~\eqref{eq:diffeq-discrete}. 

We are now in a position to apply Theorem 5.1 of~\cite{wormald},
using the quadratic function $f(\vctr{\esv}/n)$ obtained by removing
the $O(L(w)/n)$ term from the right side of~\eqref{eq:diffeq-discrete}.
That theorem has three hypotheses: a boundedness hypothesis that is
trivially satisfied because $|\esv_{\sigma}(t+1) - \esv_{\sigma}(t)| \leq 2 L(w)$,
a trend hypothesis that corresponds to Equation~\eqref{eq:diffeq-discrete},
with our error term $O(L(w)/n)$ playing the role of the term that
Wormald denotes by $\lambda_1$, and a Lipschitz hypothesis that
is satisfied because we have explicitly written our function $f$
as a sum of a bounded (i.e.\ independent of $n$) number of quadratic
functions whose coefficients are independent of $n$ as well.
Wormald's Theorem has two conclusions.  The first asserts the
existence of a unique solution to the differential equation:
for $\vctr{z}$ in the open set 
$D = \{ \vctr{\esv} \mid 0 < \esv_{\sigma} < 1 \; \forall \sigma\}$
the system of 
differential equations $\tfrac{d \vctr{z}}{dx} = f(\vctr{z})$
has a unique solution in $D$ satisfying the initial condition
$\vctr{z}(0) = \vctr{\esv}(0)/n$ and extending to points arbitrarily
close to the boundary of $D$. 
The second, more important, conclusion is that for a sufficiently
large constant $C$, with probability $1 - O(n^{1/4} \exp(-n^{1/4}))$,
$\esv_{\sigma}(t) = n z_{\sigma}(t/n) + O(n^{3/4})$
for all $t \leq \kappa n$, provided that the constant $\kappa$ is
such that the differential equation solution $\vctr{z}$ satisfies
$\min_{\sigma} \min_{0 \leq x \leq \kappa} \{z_{\sigma}(x)\} \geq C n^{-1/4}.$

We will now apply Wormald's theorem,
to prove that the numbers of unhappy
$\xx$'s and $\oo$'s 
remain nearly balanced with high probability,
by exploiting symmetry properties of the initial condition and 
of the differential equation itself.  
For a string $\sigma \in \{\xx,\oo\}^{L(w)}$ let
$\overline{\sigma}$ denote the string obtained 
from $\sigma$ by replacing every occurrence of
$\xx$ with $\oo$ and vice-versa.
For a $2^{L(w)}$-dimensional
vector $\vctr{z}$, let $\invol(\vctr{z})$ denote the
vector whose $\sigma^{\mathrm{th}}$ component is $z_{\overline{\sigma}}$
for all $\sigma$. 
It is straightforward to verify that the function $f$ defining
our differential equation satisfies $\invol(f(\vctr{z})) = 
f(\invol(\vctr{z}))$ for all $\vctr{z}$, either by analyzing
the formula~\eqref{eq:diffeq-discrete} defining $f$, or by
directly appealing to the symmetry of the rules defining the
segregation process.  Consequently, the fixed-point set 
of $\invol$ is invariant under the differential
equation: if $\vctr{z}(0)$ is a fixed point of $\invol$
then so is $\vctr{z}(x)$.  

The only remaining difficulty is that our discrete-time
process (as opposed to its continuum limit) does not start
from an initial state vector that is literally invariant
under the mapping $\invol$; it is only approximately
invariant.  The following paragraphs quantify the error
in this approximation and its contribution to the overall
error term in our application of Wormald's Theorem.

Define a function $u(\sigma)$ by specifying that
$u(\sigma)=1$ if a ring with
labeling $\sigma$ contains an unhappy $\xx$ 
in its first node, $u(\sigma)=-1$ if the first node
is an unhappy $\oo$, 
and $u(\sigma)=0$ otherwise.
Consider the linear function
\begin{align*}
\Delta(\vctr{z}) &= 
\sum_{\sigma} u(\sigma) z_\sigma .
\end{align*}
Note that $\Delta(\esv(t))$ is simply the difference between the
fraction of nodes containing unhappy individuals of type $\xx$ and $\oo$
at time $t$, 
and that $\Delta \equiv 0$ on the fixed-point set of $\invol$.
For constants $c,\kappa>0$, let $S(c,\kappa)$ denote the set of 
vectors $\vctr{z}$ such that the differential equation solution
with $\vctr{z}(0) = z$ satisfies $|\Delta(\vctr{z}(x))| < c$ 
for all $x \in [0,\kappa]$.  By the continuity properties
of ordinary differential equations (e.g., pages 135--136 
of~\cite{diffeq_book}) it follows that $S(c,\kappa)$ is an
open set.  Furthermore, it contains the set
$\Sigma = \{\vctr{z} \mid \invol(\vctr{z}) = \vctr{z}, \, 
\forall \sigma \, z_\sigma \geq 0, \, \sum_\sigma z_\sigma = 1\}.$
Since $\Sigma$ is compact, it follows that there is a constant
$\delta(c,\kappa)$ such that every point within
distance $\delta(c,\kappa)$ of $\Sigma$ belongs to $S(c,\kappa)$.

The expected squared distance from $\vctr{\esv}(0)/n$ to
$\Sigma$ is easy to bound from above, 
because $\vctr{\esv}(0)/n - \invol(\vctr{\esv}(0))/n$
is a sum of $n/L(w)$ independent random vectors (corresponding to
the $n/L(w)$ rings), each having
mean 0 and expected squared length at most $(L(w) / n)^2$.  Consequently,
by Markov's inequality,
the probability that the squared distance of $\vctr{\esv}(0)/n$ from
$\Sigma$ exceeds $\delta^2 = (\delta(c,\kappa))^2$ is at most 
$\tfrac{L(w)}{\delta^2 n}$.

Combining all of these arguments, we have the following result.

\begin{theorem} \label{thm:balance-toogeneral}
Consider the segregation process on a ring of length $n$ with
window size $w$.
Suppose $c = c(w)$ and $\kappa = \kappa(w)$ are positive 
numbers such that $c > 2 e^{12 w \kappa} \left( \tfrac{2w-1}{L(w)} \right)$. 
Then there is a positive constant
$\gamma = \gamma(c,\kappa)$ such that for all sufficiently large $n$,
with probability greater than $1 - \tfrac{3}{\gamma n}$,
at every time $0 \leq t \leq \kappa n$, 
the numbers of unhappy $\xx$'s and $\oo$'s differ by
at most $3cn$.
\end{theorem}
\begin{proof}
Set the constants $c,\kappa$ as in the above discussion,
and set $\gamma = \delta^2/L(w)$ where $\delta = \delta(c,\kappa)$
is defined above.
There are three possible reasons that the numbers of unhappy
$\xx$'s and $\oo$'s could differ by
more than $3cn$ at some time $0 \leq t \leq \kappa n$.  
\begin{enumerate*}
\item {\em The distance of the initial vector $\vctr{\sv}(0)/n$ 
from the set $\Sigma$ is greater than $\delta(c,\kappa)$.  }

We have
seen above 
that the probability of this event is at most $\tfrac{1}{\gamma n}$.
\item {\em The vector $\vctr{\sv}(0)/n$ is within distance $\delta(c,\kappa)$
of $\Sigma$, yet  at some time $0 \leq t \leq \kappa n$, 
the numbers of unhappy $\xx$'s and $\oo$'s differ
by more than $2cn$ in the segregation process on $G$.}

The assumption that $\vctr{\sv}(0)/n$ is within distance $\delta(c,\kappa)$
of $\Sigma$ implies that it belongs to $S(c,\kappa)$ and hence that the
differential equation solution satisfies $|\Delta(z(x))| < c$ for all
$x \in [0,\kappa]$, and in particular this holds when $x = t/n$.  Our 
assumption about the imbalance between unhappy $\xx$'s and $\oo$'s at
time $t$ implies that $|\Delta(\vctr{\sv}(t))/n| > 2c$ and hence that
$|\Delta(\vctr{\sv}(t)/n - z(t/n))| > c$.  Using the formula defining
$\Delta$, and recalling that $|u(\sigma)| \leq 1$ for all $\sigma$,
we now see that 
$$\sum_\sigma |\sv_\sigma(t) - n z_\sigma(t/n)| > c n$$
hence  there exists $\sigma$ such that 
$\sv_\sigma(t) - n z_{\sigma}(t/n) > \tfrac{cn}{L(w)}$.  
For sufficiently large $n$ this exceeds the $O(n^{3/4})$ 
error term in Wormald's theorem, hence the probability of this case 
occurring is $O(n^{1/4} \exp(-n^{1/4}))$, which is less than 
$\tfrac{1}{\gamma n}$ for sufficiently large $n$.
\item  {\em At some time $0 \leq t \leq \kappa n$,
the numbers of unhappy $\xx$'s and $\oo$'s
differ by at most $2cn$ in the segregation process on $G$,
yet they differ by more than $3cn$ in the segregation process on $C_n$.}

In this case, the number of tainted nodes at time $t$ must be greater
than $cn$.  
Therefore, the probability that this event occurs
at any particular time $t$ is at most 
$\exp \left( - \tfrac{wn}{36 L(w)^2} \right)$.
Taking the union bound over all $t$ in the range $0,\ldots, \kappa n$, 
we find that the probability of this case is bounded above
by $\kappa n \exp \left( - \tfrac{wn}{36 L(w)^2} \right)$,
which is again less than $\tfrac{1}{\gamma n}$ for sufficiently 
large $n$.
\end{enumerate*}
Combining the upper bounds for the 
probabilities of these three bad events
using the union bound, we obtain the 
probability estimate stated in the theorem.
\end{proof}

In order to obtain Theorem~\ref{thm:balance} from
Theorem~\ref{thm:balance-toogeneral} it is necessary
to specify values for $L(w), c(w), \kappa(w)$.  
We define
\begin{align}
\label{eq:cdef}
c(w) &= \tfrac13 w^{-4} \\
\label{eq:kdef}
\kappa(w) &= 2 w^2 \ln(w) \\
\label{eq:ldef}
L(w) &= 12 e^{12 w \kappa(w)} 
\end{align}
These choices are justified by the following proposition,
which completes the proof of Theorem~\ref{thm:balance}.

\begin{prop} \label{prop:kappa}
Define $c(w),\kappa(w),L(w)$ as in~\eqref{eq:cdef}-\eqref{eq:ldef}
above.  Let $T_0$ be the earliest time at which fewer than $3n/w^2$
individuals are impatient and let 
$T_1$ 
be the earliest time at which the numbers of 
unhappy 
$\xx$'s and 
$\oo$'s differ by more than $3 c(w) n = n/w^4$.  The
probability that 
$T_0 \leq \kappa(w) n < T_1$
is at least  $1-1/w$, for all sufficiently
large $n$.
\end{prop}
\begin{proof}
Define a \emph{pending node} to be a node that has not yet
reached its satisfaction time.
Let $\phi_\xx(t), \phi_\oo(t)$ 
denote 
the fraction of pending nodes at time $t$ 
that are labeled with $\xx, \oo$, respectively,
and let $\phi(t) = \phi_\xx(t) + \phi_\oo(t)$.
Similarly, let $\psi_\xx(t), \psi_\oo(t)$ denote
the fraction of unhappy nodes at time $t$
that are labeled with $\xx,\oo$, respectively,
and let $\psi(t) = \psi_\xx(t) + \psi_\oo(t)$.

Note that $\phi_\xx(t+1) = \phi_\xx(t)$ unless
the proposed swap at time $t$ involves a pending $\xx$
and an unhappy $\oo$, in which case 
$\phi_\xx(t+1) = \phi_\xx(t) - \tfrac{1}{n}$
since exactly one $\xx$ is no longer pending.
The probability of proposing a swap between a
pending $\xx$ and an unhappy $\oo$ at time $t$ is
$2 \phi_\xx(t) \psi_\oo(t)$.  Hence
\begin{equation} \label{eq:phixx}
\E[\phi_\xx(t+1)-\phi_\xx(t) \mid \phi_\xx(t),\psi_\oo(t)]
= - \tfrac{2}{n} \phi_\xx(t) \psi_\oo(t).
\end{equation}
Similarly,
\begin{equation} \label{eq:phioo}
\E[\phi_\oo(t+1)-\phi_\oo(t) \mid \phi_\oo(t),\psi_\xx(t)]
= - \tfrac{2}{n} \phi_\oo(t) \psi_\xx(t).
\end{equation}
Using $\phi_*(t),\psi_*(t)$ as shorthand
for the quadruple of random variables 
$(\phi_\xx(t),\phi_\oo(t),\psi_\xx(t),\psi_\oo(t))$ and summing,
we obtain
\begin{eqnarray} 
\E[\phi(t+1)-\phi(t) \mid \phi_*(t),\psi_*(t)] = 
-\tfrac{2}{n} (\phi_\xx(t) \psi_\oo(t) + \phi_\oo(t) \psi_\xx(t)).
\label{eq:phi1}
\end{eqnarray}
We can obtain an upper bound on the right side of~\eqref{eq:phi1},
provided that $t < \min\{T_0,T_1\}$.  
Indeed, for any such $t$ we have $\psi(t) \geq 3/w^2$
(since there are at least $3n/w^2$ impatient individuals
at time $t$ and all of them are unhappy),
whereas 
$$|\psi_\xx(t)-\psi_\oo(t)| \leq 1/w^4 \leq 1/w^2 \leq \psi(t)/3.$$
Since $\psi_\xx(t) + \psi_\oo(t) = \psi(t)$, it follows that
$\min\{\psi_\xx(t),\psi_\oo(t)\} \geq \psi(t)/3$ and consequently,
\begin{eqnarray}
\E[\phi(t+1)-\phi(t) \mid \phi_*(t),\psi_*(t)] \leq -\tfrac{2}{n} (\phi_\xx(t) + \phi_\oo(t)) \tfrac{\psi(t)}{3} 
\leq -\tfrac{2}{w^2 n} \phi(t),
\label{eq:phi2}
\end{eqnarray}
where the last inequality used the facts that $\psi(t) \geq 3/w^2$
and $\phi(t) = \phi_\xx(t) + \phi_\oo(t)$.  Rearranging terms
in~\eqref{eq:phi2} we obtain the bound
\begin{equation*}
E[\phi(t+1) \mid \phi_*(t),\psi_*(t)] \leq
\left( 1 - \tfrac{2}{w^2 n} \right) \phi(t)
< e^{-2/w^2 n} \phi(t),
\end{equation*}
which implies that the sequence of random variables
defined by 
$$
Y_t = \begin{cases}
e^{2t/w^2 n} \phi(t) & \mbox{if $t < \min\{T_0,T_1\}$} \\
Y_{t-1} & \mbox{otherwise}
\end{cases}
$$ 
is a supermartingale.
The initial condition $Y_0 \leq 1$ now implies that for all $t$,
$\E[Y_t] \leq 1$.

Now we specialize to a fixed value of $t$,
namely $$t=\kappa(w) n = 2 w^2 \ln(w) n.$$
If $t < \min\{T_0,T_1\}$ then
the number of impatient nodes at time $t$ is at least 
$3n/w^2$, hence $\phi(t) \geq 3/w^2$.  This implies
$$
Y_t = e^{2t/w^2 n} \phi(t) \geq w^4 \cdot (3/w^2) = 3 w^2.
$$
Using Markov's inequality,
\begin{align*}
\Pr[\kappa(w) n < \min\{T_0,T_1\}] \leq
\Pr[Y_t \geq 3 w^2] \leq \tfrac{1}{3 w^2}.
\end{align*}
Theorem~\ref{thm:balance-toogeneral} says that
for some constant $\gamma(w)$, 
$$\Pr[\kappa(w) n < T_1] \geq 1 - \tfrac{3}{\gamma(w) n}.$$
But if $\kappa(w) n < T_1$ and $\kappa(w) n \not< \min\{T_0,T_1\}$,
it means that $T_0 \leq \kappa(w) n < T_1$.
Thus, by the union bound,
\[
\Pr[T_0 \leq \kappa(w) n < T_1] \geq 1 - 
\tfrac{3}{\gamma(w) n} - \tfrac{1}{3 w^2}.
\]
For sufficiently large $n$, the right side
is greater than $1 - 1/w$, as desired.
\end{proof}
